\documentclass[a4paper,10pt]{article}
\usepackage[bbgreekl]{mathbbol}
\usepackage[small]{titlesec}
\usepackage[utf8]{inputenc}
\usepackage[T1]{fontenc}  
\usepackage[]{mdframed}
\usepackage{mathrsfs}
\usepackage[backend=biber,style=alphabetic,natbib=false,giveninits=true,doi=true,isbn=false,url=false,date=year,maxbibnames=99,sorting=nty,defernumbers=true]{biblatex}
\DeclareNameAlias{default}{family-given}

\DeclareFieldFormat[article]{title}{#1}
\renewbibmacro{in:}{}
\DeclareFieldFormat[article]{pages}{#1}
\DeclareFieldFormat{journal}{#1}
\renewcommand\bf\bfseries

\renewbibmacro*{volume+number+eid}{\printfield{volume}\setunit*{\addnbspace}\printfield{number}\setunit{\addcomma\space}\printfield{eid}}
\DeclareFieldFormat[article]{number}{\mkbibparens{#1}}
\DeclareFieldFormat[article]{volume}{\textbf{#1}}
\DeclareFieldFormat{year}{\mkbibparens{#1}}
\DeclareBibliographyDriver{article}{\printnames{author}:\newunit\newblock
	\printfield{title}\newunit\newblock
	\printfield{journaltitle}
	\newunit
	\iffieldundef{number}{\printfield{volume}}{\printfield{volume}\addspace\printfield{number}}
	\addcomma\addspace\printfield{pages}\addspace
	\printfield{year}
}
\AtEveryBibitem{\clearfield{month}}
\AtEveryBibitem{\clearfield{day}}
\usepackage[english]{babel}
\usepackage[T1]{fontenc}
\usepackage{csquotes}
\usepackage{bbm}
\usepackage[leqno]{amsmath}
\usepackage{amsfonts,amsthm,amsbsy,amssymb,dsfont,stmaryrd}
\usepackage[dvipsnames]{xcolor}
\usepackage{braket}

\makeatletter
\newcommand{\leqnomode}{\tagsleft@true\let\veqno\@@leqno}
\newcommand{\reqnomode}{\tagsleft@false\let\veqno\@@eqno}
\makeatother

\usepackage{mathtools}
\usepackage[makeroom]{cancel}

\numberwithin{equation}{section}

\newcommand\myshade{85}
\colorlet{mylinkcolor}{violet}
\colorlet{mycitecolor}{YellowOrange}
\colorlet{myurlcolor}{Aquamarine}

\usepackage[left=2.5cm,right=2.5cm,top=2.5cm,bottom=2.5cm]{geometry}
\usepackage[unicode=true,pdfusetitle,
bookmarks=true,bookmarksnumbered=false,bookmarksopen=false,
breaklinks=false,pdfborder={0 0 1},backref=false,
linkcolor  = mylinkcolor!\myshade!black,
citecolor  = mycitecolor!\myshade!black,
urlcolor   = myurlcolor!\myshade!black,
colorlinks = true,
]
{hyperref}
\usepackage[nameinlink]{cleveref}

\usepackage{tikz}
\usetikzlibrary{arrows}
\usetikzlibrary{intersections}

\usepackage{graphicx}
\usepackage{caption}
\usepackage{slashed}
\definecolor{ct_black}{HTML}{000000}
\definecolor{ct_orange}{HTML}{ED872D}
\definecolor{ct_purple}{HTML}{7A68A6}
\definecolor{ct_blue}{HTML}{348ABD}
\definecolor{ct_turquoise}{HTML}{188487}
\definecolor{ct_red}{HTML}{E32636}
\definecolor{ct_pink}{HTML}{CF4457}
\definecolor{ct_green}{HTML}{467821}

\definecolor{ct2_green}{HTML}{9FF781}
\definecolor{ct2_green_dark}{HTML}{088A08}

\theoremstyle{plain}
\newtheorem{thm}{\protect\theoremname}[section]
\theoremstyle{plain}
\newtheorem{lem}[thm]{\protect\lemmaname}
\theoremstyle{plain}

\theoremstyle{plain}
\newtheorem{prop}[thm]{\protect\propositionname}
\theoremstyle{plain}

\theoremstyle{remark}
\newtheorem{rem}[thm]{\protect\remarkname}

\theoremstyle{definition}

\theoremstyle{plain}
\newtheorem{example}[thm]{\protect\examplename}
\providecommand{\assumptionname}{Assumption}
\providecommand{\claimname}{Claim}
\providecommand{\corollaryname}{Corollary}
\providecommand{\definitionname}{Definition}
\providecommand{\lemmaname}{Lemma}
\providecommand{\propositionname}{Proposition}
\providecommand{\remarkname}{Remark}
\providecommand{\theoremname}{Theorem}
\providecommand{\examplename}{Example}

\crefname{section}{Section}{Sections}
\crefname{appendix}{Appendix}{Appendices}
\crefname{figure}{Figure}{Figures}
\crefname{assumption}{Assumption}{Assumptions}
\crefname{thm}{Theorem}{Theorems}
\crefname{lem}{Lemma}{Lemmas}
\crefname{prop}{Proposition}{Propositions}
\crefname{rem}{Remark}{Remarks}
\crefname{example}{Example}{Examples}
\crefformat{equation}{(#2#1#3)}
\crefname{table}{Table}{Tables}

\crefrangelabelformat{equation}{(#3#1#4--#5#2#6)}

\crefmultiformat{equation}{(#2#1#3}{, #2#1#3)}{#2#1#3}{#2#1#3}
\Crefmultiformat{equation}{(#2#1#3}{, #2#1#3)}{#2#1#3}{#2#1#3}

\newtheorem*{lem*}{\protect\lemmaname}

\newcommand{\ee}{\operatorname{e}}
\newcommand{\ii}{\operatorname{i}}

\newcommand{\ZZ}{\mathbb{Z}}
\newcommand{\TT}{\mathbb{T}}

\newcommand{\NN}{\mathbb{N}}
\newcommand{\RR}{\mathbb{R}}

\newcommand{\calF}{\mathcal{F}}

\newcommand{\calO}{\mathcal{O}}
\newcommand{\calR}{\mathcal{R}}

\newcommand{\calH}{\mathcal{H}}

\newcommand{\calJ}{\mathcal{J}}

\newcommand\norm[1]{\left\lVert#1\right\rVert}

\newcommand{\dif}{\operatorname{d}\!} \newcommand{\tr}{\operatorname{tr}}

\newcommand{\Id}{\mathds{1}}

\newcommand{\sgn}{\operatorname{sgn}}

\newcommand{\supp}{\operatorname{supp}}

\usepackage{environ}

\NewEnviron{malign}{\begin{align}\begin{split}
			\BODY
	\end{split}\end{align}
}

\newcommand{\eq}[1]{\begin{align*}#1\end{align*}}
\newcommand{\eql}[1]{\begin{align}#1\end{align}}

\usepackage{tikz}

\title{Euler Topological Metals in 1D}
\author{\href{mailto:jc1220@math.princeton.edu}{Yichen Hu}\\
	{\footnotesize Department of Physics, Florida Atlantic University}\\
	\href{mailto:shapiro@math.princeton.edu}{Jacob Shapiro}\\
	{\footnotesize Department of Mathematics, Princeton University}
}

\begin{document}
	\reqnomode
	
	\maketitle
	\begin{abstract}
		We give a rigorous one-dimensional formulation of Kane's transport proposal for probing the Euler characteristic of a Fermi sea. For a clean, translation-invariant continuum Hamiltonian with real-analytic dispersion, a Fermi sea confined to a finite momentum range, and nonzero Fermi velocity at every Fermi point, we analyze a particular regularized transport response (obtained by tracing over the Fermi sea) and prove that its large-time limit equals the Euler characteristic. We derive an explicit finite-time formula, show that quantization requires the large-time limit, and obtain a convergence rate under additional nonstationarity assumptions. We also prove the lattice analog, where the commutator of the spatial switch with its time evolution is trace-class and a completely filled band contributes zero. Finite-volume calculations illustrate the prescribed order of the thermodynamic and large-time limits.
	\end{abstract}
	\section{Introduction}
    Topological phenomena in condensed matter physics are most commonly
    associated with gapped phases, where an invariant of the occupied Bloch
    bundle, such as a Chern number, is detected by a quantized transport
    coefficient \cite{Hasan_Kane_2010}. Metals have no such spectral gap, but
    they possess another natural geometric object: the occupied region in
    momentum space. Changes in the topology of this region are the classical
    Lifshitz transitions, whose thermodynamic and kinetic signatures have long
    been studied \cite{Lifshitz1960,BlanterEtAl1994}.

    Kane proposed a direct transport probe of Fermi-sea topology in his 2021
    preprint, highlighted at the
    time by Beenakker \cite{Beenakker2021Euler}, and subsequently published in
    \cite{Kane_22_PhysRevLett.128.076801}. The proposal extends the familiar
    one-dimensional Landauer conductance \cite{Landauer1957}. For a one-dimensional continuum Fermi sea, or a proper partially occupied
    lattice band, the number of ballistic channels is the number of components
    of the Fermi sea and hence its Euler characteristic. A completely filled
    periodic band has no Fermi channel and contributes zero. The use of the term \emph{Euler} in this paper refers to the Euler characteristic of the occupied momentum-space region and is distinct from the Euler-class topology of real Bloch-band subspaces studied in the literature on Euler semimetals \cite{JankowskiEtAl2024EulerSemimetals}. Kane showed that in two
    dimensions the corresponding invariant is obtained instead from a
    second-order, three-terminal response: electron-like Fermi surfaces and
    hole-like Fermi surfaces enter with opposite signs. More generally, the
    proposal suggests a hierarchy in which a $d$-dimensional Fermi sea is
    probed by a $d$th-order response of $d+1$ regions meeting at a point. The
    cases $d=1,2$ were established in \cite{Kane_22_PhysRevLett.128.076801};
    the higher-dimensional transport problem was left open.

    Several subsequent works developed different aspects of this idea. Yang
    and Zhai proposed a two-dimensional ultracold-atom implementation and
    quantified the effects of finite pulses, temperature, and trap geometry
    \cite{YangZhai2022}. Zhang formulated a one-dimensional linear response in
    a trapped gas and its nonlinear higher-dimensional analogues
    \cite{Zhang2023}. Tam, Claassen, and Kane found that the same Euler
    characteristic controls a multipartite entanglement quantity associated
    with $d+1$ regions: for odd $d$ this is a multipartite mutual information,
    while for even $d$ a charge-weighted version is required
    \cite{TamClaassenKane2022}. In a different direction, Andreev bound states
    in a long Josephson junction furnish a rectified transport probe in two
    dimensions \cite{TamKane2023Andreev,TamDeBeuleKane2023}, and the Fermi-sea
    Euler characteristic can be inherited by the invariant of a weak-pairing
    topological superconductor \cite{YangLiLi2023,Jia2025}.

    Equal-time correlations provide a further formulation which is not tied
    to transport. The connected $(d+1)$-point density correlation of a free
    Fermi gas has a universal long-wavelength term proportional to the Euler
    characteristic \cite{TamKane2024}. This prediction was recently tested in
    a two-dimensional gas of ${}^{6}\mathrm{Li}$ using single-atom-resolved
    measurements of three- and four-point correlations
    \cite{DaixEtAl2025}. The limitations of the quantization are now also
    becoming visible: Berry curvature preserves the homogeneous
    two-dimensional transport formula but, together with spatial
    inhomogeneity, produces a non-quantized correction \cite{YangLi2026}; and
    in the long-wavelength collinear limit, the singular three-point
    correlation in an interacting two-dimensional Fermi liquid survives with
    a coefficient renormalized by the Landau parameters \cite{TamKane2026,Kane2026ThreePoint}.

    These developments show that the Euler characteristic is not exclusively
    a transport invariant. It may appear in nonlinear response, in equal-time
    density correlations, in bounds on charge fluctuations
    \cite{TamEtAl2026ChargeFluctuations}, or in specially arranged entanglement quantities.
    The last qualification is important. In dimensions $d>1$, the leading
    bipartite entanglement entropy of a Fermi sea is governed by a Widom-type
    geometric integral over the real-space boundary and the Fermi surface,
    conjectured in \cite{GioevKlich2006} and proved in
    \cite{LeschkeSobolevSpitzer2014}; it does not in general isolate the Euler
    characteristic. In one dimension, by contrast, the leading logarithmic
    coefficient already counts the Fermi intervals. The multipartite
    construction of \cite{TamClaassenKane2022} extracts a topological term,
    with charge weighting in even dimensions.

    The purpose of the present paper is narrower. We give a mathematical
    formulation of the one-dimensional continuum response underlying Kane's
    argument. A continuous-translation-invariant single-particle Hamiltonian
    on $L^2(\RR^d)$ is of the form
    \eq{H=f(P),\qquad P=-\ii\nabla\,,}
    and at Fermi energy $E_F$ its Fermi sea is
    \eq{\Sigma_{E_F}:=\Set{p\in\RR^d\mid f(p)\leq E_F}\,.}
    We assume below that the one-dimensional dispersion $f:\RR\to\RR$ is real analytic, that
    $E_F$ is a regular value, and that $\Sigma_{E_F}$ is compact. In one
    dimension $\chi(\Sigma_{E_F})=b_0(\Sigma_{E_F})$ is simply the number of
    its connected components. In two dimensions it is
    $b_0(\Sigma_{E_F})-b_1(\Sigma_{E_F})$, where $b_1(\Sigma_{E_F})$ is the number of holes. For general $d$ one has \eql{\chi(\Sigma_{E_F}) = \sum_{j=0}^d (-1)^{j}b_j(\Sigma_{E_F})} where $b_j(\Sigma_{E_F})$ is the $j$th Betti number of the Fermi sea.

    Even in one dimension, the formal Kubo expression contains sharp spatial
    switches and is not an ordinary trace in infinite volume. Our main result,
    \cref{thm:main theorem}, identifies a regularized linear response with
    an ordinary sharp-switch trace whose large-time limit equals
    $\chi(\Sigma_{E_F})$. We also quantify why the large-time limit is
    necessary and give the analogous argument for a lattice Hamiltonian in
    \cref{sec:discrete}.

    The higher-dimensional questions remain substantial. In two dimensions
    one would like a rigorous nonlinear Kubo formula for three intersecting
    spatial switches, with a controlled infinite-volume regularization and
    order of limits. In three dimensions it is not yet clear which transport
    geometry produces the Euler characteristic, or whether a genuinely linear
    response with a more elaborate observable can replace the expected
    third-order response. More generally, it remains to determine which of the
    transport, density-correlation, and multipartite-entanglement formulations
    survive disorder and interactions, and whether their quantized terms are
    different realizations of a single index theorem.

    This paper is organized as follows.  In \cref{sec:continuum} we derive the formal Kubo response, identify the trace-class obstruction in the continuum, and state our regularized formulation.  After giving semiclassical and Boltzmann motivations, we prove the main theorem by direct momentum-space analysis and quantify the approach to the quantized large-time limit.  In \cref{sec:discrete} we establish the corresponding lattice result, where the commutator of the spatial switch with its time evolution is already trace-class and a completely filled band contributes zero.  Finally, \cref{sec:numerics} illustrates the finite-time and finite-volume effects, the role of boundary conditions, and the behavior of the response for the one-dimensional Anderson model.
\section{The one-dimensional theory}\label{sec:continuum}
    
    \subsection{The experimental setup--formal manipulations}

Consider a non-interacting electron in a continuous one-dimensional infinite
wire, with single-particle Hilbert space $L^2(\RR)$. We work in units
$\hbar=1$ and write $-e$, with $e>0$, for the electric charge of the
electron. Let
\eq{\Lambda:=\chi_{(0,\infty)}} where $\chi_S$ is the characteristic function of the set $S\subseteq\RR$. 

At time $t_0\in\RR$, we apply the electric potential impulse
\eq{\phi(x,t)=-\frac{E_0}{e}\Lambda(x)\delta(t-t_0),
\qquad E_0>0\,.}
The corresponding electric field is
\eq{\mathcal{E}(x,t)
=-\partial_x\phi(x,t)
=\frac{E_0}{e}\delta(x)\delta(t-t_0)}
and the associated time-dependent Hamiltonian (free dynamics plus the impulse perturbation) is
\eq{H(t)=H+E_0\Lambda(X)\delta(t-t_0)\,.}

The observable
\eq{N_R:=\Lambda(X)}
measures the number of particles on the right half-line. At zero temperature, the unperturbed one-particle
density matrix $\rho$ is the Fermi projection
\eq{P_F:=\chi_{(-\infty,E_F)}(H)\,.}

In finite volume, Kubo's formula
\cite{Kubo_57_I_doi:10.1143/JPSJ.12.570,
Kubo_57_II_doi:10.1143/JPSJ.12.1203,Bruus2004-lk}
gives, to first order in $E_0$,
\eql{\tr\left((\rho(t)-P_F)B\right)
&=-\ii E_0\int_{-\infty}^{t}\dif{t'}\,
\tr\left(
\ee^{-\ii(t-t')H}
\left[\Lambda(X)\delta(t'-t_0),P_F\right]
\ee^{\ii(t-t')H}B
\right)
+\calO(E_0^2)\notag\\
&=-\ii E_0\chi_{(0,\infty)}(t-t_0)
\tr\left(
\ee^{-\ii(t-t_0)H}
[\Lambda(X),P_F]
\ee^{\ii(t-t_0)H}B
\right)
+\calO(E_0^2)\,.
\label{eq:Kubo}}

Taking $B=N_R=\Lambda(X)$ and using $[H,P_F]=0$, formal cyclicity of the
trace gives
\eq{\tr\left(
\ee^{-\ii(t-t_0)H}
[\Lambda(X),P_F]
\ee^{\ii(t-t_0)H}\Lambda(X)
\right)
=
-\tr\left(
P_F\left[
\ee^{-\ii(t-t_0)H}\Lambda(X)\ee^{\ii(t-t_0)H},
\Lambda(X)
\right]
\right).}
Consequently, the formal particle-number response is
\eql{\Delta N_R(t)
&:=\tr\left((\rho(t)-P_F)\Lambda(X)\right)\notag\\
&=
\ii E_0\chi_{(0,\infty)}(t-t_0)
\tr\left(
P_F\left[
\ee^{-\ii(t-t_0)H}\Lambda(X)\ee^{\ii(t-t_0)H},
\Lambda(X)
\right]
\right)
+\calO(E_0^2)\,.
\label{eq:putative starting point}}

The electric-charge response has the opposite sign because $Q_R=-eN_R$:
\eq{\frac{2\pi}{eE_0}\Delta Q_R(t)
=
-2\pi\ii\chi_{(0,\infty)}(t-t_0)
\tr\left(
P_F\left[
\ee^{-\ii(t-t_0)H}\Lambda(X)\ee^{\ii(t-t_0)H},
\Lambda(X)
\right]
\right)
+\calO(E_0)\,.}
This is the origin of both the sign and the factor $2\pi$ in the response
considered below.

\subsection{The main theorem}

From now on we take $t_0=0$, write $\Lambda$ for $\Lambda(X)$, and set
\eq{\Lambda_t:=\ee^{\ii tH}\Lambda\ee^{-\ii tH}.}
We denote the trace-class ideal by $\calJ_1$ and its norm by
$\norm{\cdot}_1$.

    \begin{thm}[The response is the Euler characteristic]\label{thm:main theorem}
        Let $H=f(P)$, where $f:\RR\to\RR$ is real analytic. Assume that $E_F$ is a regular value of $f$ and that the Fermi sea
        \eql{\label{eq:Fermi sea}\Sigma_{E_F}=\Set{p\in\RR | f(p)\leq E_F}}

        is compact. Then $P_F[\Lambda_t,\Lambda]\in\calJ_1$ for every
        $t\geq0$, and the dimensionless response equals the Euler characteristic

        \eql{\lim_{t\to\infty}2\pi\ii\tr\left(P_F[\Lambda_t,\Lambda]\right)
=\chi(\Sigma_{E_F})\,.
\label{eq:linear response}}

    \end{thm}
The large-time limit is essential: the finite-time response is not
topologically quantized; see \cref{prop:finite-time response} below.

The formal manipulations above require care in infinite volume.
The commutator $[\Lambda(X),P_F]$ need not be trace-class.
For a bounded finite-range one-dimensional lattice Hamiltonian with
$E_F$ in a spectral gap, $P_F$ is a smooth function of $H$ and
$[\Lambda(X),P_F]$ is trace-class. This argument is unavailable without
a gap.

In the continuum, $[\Lambda(X),H]$ need not be bounded, and even
commutators with smooth bounded functions of $H$ can fail to be
trace-class.

To relate the trace in the theorem to Kubo's formula, we smooth the
occupation and leave the spatial switch sharp.
Let $g:\RR\to[0,1]$ be smooth and nonincreasing, with
$g(s)=1$ for $s\leq-1$ and $g(s)=0$ for $s\geq0$, and set
\eq{\rho_\eta(E):=g\!\left(\frac{E-E_F}{\eta}\right),\qquad
r_\eta(p):=\rho_\eta(f(p)),\qquad \eta>0.}
Then $\rho_\eta(H)=r_\eta(P)$, and
$\supp r_\eta\subseteq\Sigma_{E_F}$ implies
$r_\eta\in C_c^\infty(\RR)$.

\begin{lem}\label[lem]{lem:smooth multiplier trace class}
For every $r\in C_c^\infty(\RR)$,
$[\Lambda(X),r(P)]\in\calJ_1$. In particular,
$[\Lambda(X),\rho_\eta(H)]\in\calJ_1$ for every $\eta>0$.
\end{lem}
\begin{proof}
For $z=a+\ii b$, $b\neq0$, the resolvent of $P=-\ii\partial_x$ gives
\eq{[\Lambda(X),(P-z\Id)^{-1}](x,y)
=\ii\ee^{\ii z(x-y)}\chi_{\{bx>0\}}\chi_{\{by<0\}}.}
This is a rank-one kernel, with
\eq{\norm{[\Lambda(X),(P-z\Id)^{-1}]}_1=\frac{1}{2|b|}.}
Choose a compactly supported almost-analytic extension $\widetilde r$ of
$r$ with $\partial_{\bar z}\widetilde r=\calO(|b|^2)$ near $b=0$, where
$\partial_{\bar z}=\frac12(\partial_a+\ii\partial_b)$.
The Helffer--Sj\"ostrand formula, as in \cite[Appendix (A.10)]{Elbau_Graf_2002},
then yields
\eq{[\Lambda(X),r(P)]
=\frac{1}{\pi}\int_{\RR^2}\dif{a}\dif{b}\,
\partial_{\bar z}\widetilde r(z)[\Lambda(X),(P-z\Id)^{-1}].}
The displayed resolvent bound makes this integral convergent in trace
norm. Apply it to $r=r_\eta$, using $\rho_\eta(H)=r_\eta(P)$.
\end{proof}

The exact impulse and subsequent free evolution give

\eq{\rho_{\eta,E_0}(t):=
\ee^{-\ii tH}\ee^{-\ii E_0\Lambda(X)}\rho_\eta(H)
\ee^{\ii E_0\Lambda(X)}\ee^{\ii tH},\qquad t\geq0.}

By \cref{lem:smooth multiplier trace class}, $[\Lambda,\rho_\eta(H)]$
is trace-class. Integrating the derivative of
$\ee^{-\ii s\Lambda}\rho_\eta(H)\ee^{\ii s\Lambda}$ therefore gives the
exact identity in trace norm
\eq{\rho_{\eta,E_0}(t)-\rho_\eta(H)
&=-\ii\ee^{-\ii tH}\left(\int_0^{E_0}\dif{s}\,
\ee^{-\ii s\Lambda}[\Lambda,\rho_\eta(H)]
\ee^{\ii s\Lambda}\right)\ee^{\ii tH},}
and hence
\eq{\norm{\rho_{\eta,E_0}(t)-\rho_\eta(H)}_1
\leq |E_0|\norm{[\Lambda,\rho_\eta(H)]}_1.}
For every bounded observable $B$, the regularized Kubo formula is thus an
ordinary trace identity:
\eq{\tr\left((\rho_{\eta,E_0}(t)-\rho_\eta(H))B\right)
&=-\ii E_0\tr\left(\ee^{-\ii tH}[\Lambda,\rho_\eta(H)]
\ee^{\ii tH}B\right)\\
&\quad+\calO_{\eta,B}(E_0^2).}
Here $\eta>0$ is fixed; no remainder bound uniform as $\eta\to0^+$ is
asserted.

Hence the induced electric charge
\eq{\Delta Q_\eta(E_0,t):=-e\tr\left(
(\rho_{\eta,E_0}(t)-\rho_\eta(H))\Lambda(X)\right)}
is well-defined, and its normalized linear response is
\eql{\calR_\eta(t)
&:=\frac{2\pi}{e}\left.\frac{\partial}{\partial E_0}
\Delta Q_\eta(E_0,t)\right|_{E_0=0}\notag\\
&=2\pi\ii\tr\left([\Lambda(X),\rho_\eta(H)]
\ee^{\ii tH}\Lambda(X)\ee^{-\ii tH}\right).
\label{eq:regularized Kubo response}}

Strong convergence $\rho_\eta(H)\to P_F$ does not justify substituting
$P_F$ in \cref{eq:regularized Kubo response}: the resulting operator
$[\Lambda,P_F]\Lambda_t$ need not be trace-class. The lemma below
identifies the limit by first rewriting the regularized trace.

\begin{lem}\label[lem]{lem:sharp response}
Under the hypotheses of \cref{thm:main theorem}, for every $t\geq0$, $P_F[\Lambda_t,\Lambda]\in\calJ_1$, and
\eq{\lim_{\eta\to0^+}\calR_\eta(t)
=2\pi\ii\tr\left(P_F[\Lambda_t,\Lambda]\right)
=2\pi\ii\tr\left(P_F[\Lambda_t,\Lambda]P_F\right).}
\end{lem}
\begin{proof}
Choose a real $\varphi\in C_c^\infty(\RR)$ equal to $1$ on a neighborhood of
$\Sigma_{E_F}$. Since both $\varphi$ and $\varphi\ee^{\ii tf}$ are smooth and
compactly supported, \cref{lem:smooth multiplier trace class} gives
\eq{\varphi(P)(\Lambda_t-\Lambda)
=[\varphi(P)\ee^{\ii tH},\Lambda]\ee^{-\ii tH}
-[\varphi(P),\Lambda]\in\calJ_1.}
Consequently,
\eq{\varphi(P)[\Lambda_t,\Lambda]
&=[\varphi(P)(\Lambda_t-\Lambda),\Lambda]\\
&\quad-[\varphi(P),\Lambda](\Lambda_t-\Lambda)\in\calJ_1.}
As $P_F\varphi(P)=P_F$ and $\rho_\eta(H)\varphi(P)=\rho_\eta(H)$, this proves
trace-classness after multiplication by either $P_F$ or $\rho_\eta(H)$.
Moreover, $\rho_\eta(H)(\Lambda_t-\Lambda)$ is trace-class, as is
\eq{\rho_\eta(H)\Lambda(\Lambda_t-\Lambda)
=\rho_\eta(H)(\Lambda_t-\Lambda)\Lambda
-\rho_\eta(H)[\Lambda_t,\Lambda].}
The operator $[\Lambda,\rho_\eta(H)]\Lambda$ is off-diagonal relative to
$\Lambda$, so its trace is zero. Trace cyclicity therefore gives
\eq{\calR_\eta(t)
&=2\pi\ii\tr\left([\Lambda,\rho_\eta(H)](\Lambda_t-\Lambda)\right)\\
&=2\pi\ii\tr\left(\rho_\eta(H)[\Lambda_t,\Lambda]\right).}
Finally, $\rho_\eta(H)\to P_F$ strongly and $0\leq\rho_\eta(H)\leq\Id$.
Since $\varphi(P)[\Lambda_t,\Lambda]$ is trace-class,
\eq{\norm{(\rho_\eta(H)-P_F)[\Lambda_t,\Lambda]}_1
=\norm{(\rho_\eta(H)-P_F)\varphi(P)[\Lambda_t,\Lambda]}_1
\longrightarrow0.}
This proves the first trace identity in the limit. The second follows by
cyclicity and $P_F^2=P_F$.
\end{proof}

Thus, for every fixed $t\geq0$,
\eql{\lim_{\eta\to0^+}\frac{2\pi}{e}
\left.\frac{\partial}{\partial E_0}\Delta Q_\eta(E_0,t)\right|_{E_0=0}
=2\pi\ii\tr\left(P_F[\Lambda_t,\Lambda]\right)=:\calR(t).
\label{eq:Kubo sharp limit}}
The derivative is taken at fixed $\eta>0$, and the occupation cutoff is
removed at fixed time after rewriting the trace. The limit is independent
of $g$ within the class specified above. From now on, $\calR(t)$ denotes
the ordinary sharp trace.

\begin{rem}[Trace class in the continuum]
\label[rem]{rem:continuum trace class}
In the continuum, smoothness of the dispersion relation is not enough to
make the relevant commutator trace-class. What matters is also its behavior
over the entire momentum axis.

A theorem of Peller
\cite[Theorem~5.7]{Peller2024Besov}, applied after Fourier transform,
has the following simple consequence: if
\eq{
\left[\Lambda(X),\ee^{-\ii t f(P)}\right]\in\calJ_1
}
for some fixed $t\neq 0$, then the phase
$p\mapsto\ee^{-\ii t f(p)}$ must have finite total variation. For a
differentiable real-valued dispersion, this requires
\eq{
\int_{\RR}\dif{p}\,
\left|\partial_p\ee^{-\ii t f(p)}\right|
=
|t|\int_{\RR}\dif{p}\,|f'(p)|
<\infty.
}
This condition already rules out the usual proper dispersions satisfying
$f(p)\to+\infty$ as $p\to\pm\infty$: such a function necessarily has
infinite total variation. This remains true no matter how smooth or
analytic $f$ is.

For example, for the free-particle dispersion $f(p)=p^2$,
\eq{
\int_{\RR}\dif{p}\,
\left|\partial_p\ee^{-\ii t p^2}\right|
=
2|t|\int_{\RR}\dif{p}\,|p|
=
\infty.
}
Consequently, for every $t\neq 0$,
\eq{
\left[\Lambda(X),\ee^{-\ii tP^2}\right]\notin\calJ_1.
}

Thus the standard locality argument based on the trace-classness of
$\left[\Lambda(X),\ee^{-\ii tH}\right]$ is unavailable for the usual
continuum Hamiltonians. The Fermi-localized commutator
$P_F[\Lambda_t,\Lambda]$ nevertheless remains trace-class.

For completeness, Peller's exact criterion in the present setting is
\eql{
\left[\Lambda(X),\ee^{-\ii t f(P)}\right]\in\calJ_1
\quad\Longleftrightarrow\quad
\int_{\RR}\frac{\dif{h}}{|h|^2}
\int_{\RR}\dif{p}\,
\left|
\ee^{-\ii t f(p+2h)}
-2\ee^{-\ii t f(p+h)}
+\ee^{-\ii t f(p)}
\right|
<\infty.
\label{eq:Peller criterion}
}
The main point is that this is a global condition in momentum, rather
than a local smoothness condition.

A convenient, although stronger than necessary, sufficient assumption is
that, for some constant $c$,
\eq{
\int_{\RR}
\left(
|f(p)-c|+|f'(p)|+|f''(p)|
\right)\dif{p}
<\infty.
}
Equivalently, $f-c\in W^{2,1}(\RR)$. Under this assumption,
$\ee^{-\ii t f}-\ee^{-\ii t c}\in W^{2,1}(\RR)$, and Peller's criterion
shows that the commutator is trace-class for every fixed $t$.
\end{rem}

Before proving the large-time limit, we give two semiclassical motivations.

\begin{proof}[Semiclassical motivation]
Suppose that
\eq{\Sigma_{E_F}=\bigsqcup_{j=1}^{N}[a_j,b_j]\,.}
The Hamiltonian flow generated by $f(p)$ is
\eq{(x,p)\longmapsto\left(x+t f'(p),p\right)\,,}
so the formal principal symbol of
$\ee^{\ii tH}\Lambda(X)\ee^{-\ii tH}$ is
$\Lambda(x+t f'(p))$. Using
\eq{\{a,b\}=\partial_xa\,\partial_pb-\partial_pa\,\partial_xb}
and the semiclassical correspondence $[A,B]\sim\ii\{a,b\}$ gives
\eq{[\Lambda(X),\rho_\eta(H)]\sim
\ii\delta_0(x)r_\eta'(p).}
The semiclassical trace rule therefore yields, for $t>0$,
\eq{\calR_\eta(t)
&\sim-\int_\RR\dif{p}\int_\RR\dif{x}\,
\delta_0(x)r_\eta'(p)\Lambda(x+t f'(p))\\
&=-\int_\RR\dif{p}\,r_\eta'(p)\Lambda(t f'(p)).}
Since $E_F$ is a regular value,
\eq{f'(a_j)<0<f'(b_j).}
For sufficiently small $\eta$, the derivative $r_\eta'$ is supported in
disjoint neighborhoods of these endpoints on which the velocity has the
corresponding constant sign. The cutoff rises from $0$ to $1$ near each
$a_j$ and falls from $1$ to $0$ near each $b_j$. Thus only the latter
neighborhoods contribute, and the displayed semiclassical expression
is $N$.

This is not a proof: the sharp switch has a singular symbol, and the
semiclassical trace approximation is not controlled in this regime.
In particular, it does not establish finite-time quantization. The actual
proof below obtains the integer only after the large-time limit.
\end{proof}

    \subsection{Kane's Heuristics via the Boltzmann equation}
    The collisionless Boltzmann equation gives a semiclassical explanation of the normalization in \cref{eq:linear response}. We work in units $\hbar=1$, so that $h=2\pi$, and take $e>0$ to be the magnitude of the electron charge. Consider the idealized voltage pulse

    \eq{E(x,t)=\xi\frac{2\pi}{e}\delta(t)\delta(x)\,.}
    We compute the coefficient linear in $\xi$; $h/e$ fixes its normalization.

    To avoid confusing the occupation function with the dispersion $f$, write

    \eql{\label{eq:Perb}n(x,p,t)=n_0(p)+\delta n(x,p,t),\qquad n_0(p)=r_\eta(p)\,.}

    The group velocity is $v_p=f'(p)$. Since the electron charge is $-e$, the collisionless Boltzmann equation is
    \eql{\label{eq:cBe}\left(\partial_t+v_p\partial_x-eE(x,t)\partial_p\right)n(x,p,t)=0\,.}
    Substituting \cref{eq:Perb} into \cref{eq:cBe} and retaining terms linear in the field gives
    \eq{\left(\partial_t+v_p\partial_x\right)\delta n(x,p,t)=eE(x,t)\partial_p n_0(p)\,.}
    The retarded solution is therefore

    \eq{\delta n(x,p,t)&=e\int_{-\infty}^{t}\dif{t'}\int_{\RR}\dif{x'}E(x',t')\delta\left(x-x'-v_p(t-t')\right)\partial_p n_0(p)\\
    &=2\pi\xi\int_{-\infty}^{t}\dif{t'}\delta(t')\delta\left(x-v_p(t-t')\right)\partial_p n_0(p)\,.}

    For $t>0$,
    \eq{\int_0^\infty\dif{x}\int_{-\infty}^{t}\dif{t'}\delta(t')\delta\left(x-v_p(t-t')\right)=\Lambda(v_p)\,.}
    Hence the excess electric charge on the half-line $x>0$ is

    \eq{Q(t)&=-e\int_{\RR}\frac{\dif{p}}{2\pi}\int_0^\infty\dif{x}\,\delta n(x,p,t)\\
    &=-e\xi\int_{\RR}\dif{p}\,\Lambda(v_p)\partial_p n_0(p)\,.}

    If $\Sigma_{E_F}=\bigsqcup_{i=1}^{N}[a_i,b_i]$, then, distributionally,
    \eq{\partial_p n_0(p)\longrightarrow\sum_{i=1}^{N}\left(\delta(p-a_i)-\delta(p-b_i)\right)
    \quad\text{as }\eta\to0^+\,.}

    Since $v_{a_i}<0<v_{b_i}$, we conclude that

    \eq{\lim_{\eta\to0^+}\left.\frac1e\frac{\partial Q(t)}{\partial\xi}\right|_{\xi=0}
    =\sum_{i=1}^{N}\left(\Lambda(v_{b_i})-\Lambda(v_{a_i})\right)=N\,.}
    Thus the semiclassical linear-response coefficient is the Euler characteristic.
    This is Kane's original argument. It is not the response to a finite sharp
    impulse of area $h/e$, whose exact kick is $\ee^{-2\pi\ii\Lambda(X)}=I$.

\subsection{The proof of the main theorem}

We first record the elementary geometry of the Fermi sea.

\begin{lem}\label[lem]{lem:Fermi intervals}
Under the hypotheses of \cref{thm:main theorem}, either $\Sigma_{E_F}$ is empty or
\eq{\Sigma_{E_F}=\bigsqcup_{i=1}^{N}[a_i,b_i]}
with $a_i<b_i<a_{i+1}$ and
\eq{f'(a_i)<0<f'(b_i)\,.}
In particular, $N$ is the Euler characteristic of $\Sigma_{E_F}$.
\end{lem}

\begin{proof}
The boundary of $\Sigma_{E_F}$ is contained in $f^{-1}(E_F)$. Since $E_F$ is a regular value, these points are isolated, and compactness makes their number finite. Every connected component is therefore a non-degenerate compact interval. At a left endpoint $f-E_F$ crosses from positive to negative, while at a right endpoint it crosses from negative to positive, which gives the asserted signs.
\end{proof}

In momentum space, the response takes the following form.

\begin{lem}\label[lem]{lem:trace formula}
For $t>0$ and $\Sigma_{E_F}\neq\varnothing$,
\eq{&2\pi\ii\tr\left(P_F[\Lambda_t,\Lambda]\right)\\
&\qquad=-\frac{1}{\pi}\sum_{i=1}^{N}
\int_{p\in[a_i,b_i]}\dif{p}\int_{k\in[a_i,b_i]^c}\dif{k}
\frac{\sin\left(t(f(p)-f(k))\right)}{(p-k)^2}\,.}
\end{lem}

\begin{proof}
By \cref{lem:sharp response}, it suffices to compute the limit of
$2\pi\ii\tr([\Lambda,\rho_\eta(H)]\Lambda_t)$ as $\eta\to0^+$.
The commutator $[\Lambda,\rho_\eta(H)]$ is trace-class by
\cref{lem:smooth multiplier trace class}. Its diagonal blocks relative to
$\Lambda$ vanish, so
\eq{\tr\left([\Lambda,\rho_\eta(H)]\Lambda_t\right)
=\tr\left([\Lambda,\rho_\eta(H)]
(\Lambda_t-\Lambda)\right).}
In momentum space, $[\Lambda,\rho_\eta(H)]$ and
$\Lambda_t-\Lambda$ have kernels, respectively,
\eq{\frac{r_\eta(k)-r_\eta(p)}{2\pi\ii(p-k)},\qquad
\frac{\ee^{\ii t(f(p)-f(k))}-1}{2\pi\ii(p-k)},}
with diagonal singularities removed by continuity. In the second kernel,
the delta terms cancel and the numerator removes the principal value.

Compress both operators on each side by $\chi_{[-M,M]}(P)$.
Their compressed kernels are Hilbert--Schmidt. As $M\to\infty$, the
compressed commutator converges in trace norm, while the other compressed
operator converges strongly and remains uniformly bounded. Consequently
the product traces converge to the trace displayed above.
The kernel product is absolutely integrable: near the diagonal both
numerators are $\calO(|p-k|)$ on a fixed compact set, and away from it
compact support of $r_\eta$ gives an integrable quadratic tail.
The Hilbert--Schmidt trace formula, followed by $M\to\infty$ and
antisymmetry of the cosine term, therefore gives
\eq{&2\pi\ii\tr\left([\Lambda,\rho_\eta(H)]\Lambda_t\right)\\
&\qquad=-\frac{1}{2\pi}\int_{\RR^2}\dif{p}\dif{k}\,
\frac{(r_\eta(p)-r_\eta(k))
\sin\left(t(f(p)-f(k))\right)}{(p-k)^2}\,.}

For small $\eta$, the transition regions of $r_\eta$ lie in finitely many
monotone neighborhoods of the Fermi points. Changing variables from $p$
to $f(p)$ there shows that $\norm{r_\eta'}_{L^1}$ is uniformly bounded.
Since $r_\eta$ also has uniformly bounded $L^1$ norm and common compact
support,
\eq{\int_{\RR}\dif{p}\,|r_\eta(p)-r_\eta(p-h)|
\leq C\min\{|h|,1\}\,.}
For $|h|\leq1$, local boundedness of $f'$ and $|\sin s|\leq|s|$ bound
\eq{\int_{\RR}\dif{p}\,
\frac{|r_\eta(p)-r_\eta(p-h)|\,
|\sin(t(f(p)-f(p-h)))|}{h^2}}
by $Ct$, uniformly in $\eta$; for $|h|>1$ it is bounded by $C/h^2$.
As $r_\eta\to\chi_{\Sigma_{E_F}}$ almost everywhere, dominated convergence
first in $p$ for fixed $h\neq0$, and then in $h$, gives
\eq{2\pi\ii\tr\left(P_F[\Lambda_t,\Lambda]\right)
=-\frac{1}{\pi}\int_{p\in\Sigma_{E_F}}\dif{p}
\int_{k\notin\Sigma_{E_F}}\dif{k}\,
\frac{\sin\left(t(f(p)-f(k))\right)}{(p-k)^2}\,.}
The cross terms $[a_i,b_i]\times[a_j,b_j]$, $i\neq j$, cancel by
antisymmetry, yielding the asserted formula.
\end{proof}

We shall use the following elementary form of the Riemann--Lebesgue lemma.

\begin{lem}\label[lem]{lem:oscillatory integral}
Let $D\subseteq\RR^2$ be measurable, suppose that its first coordinate ranges in a compact interval and that $|p-k|\geq\eta>0$ on $D$. Then
\eq{\lim_{t\to\infty}\int_D\dif{p}\dif{k}\frac{\sin\left(t(f(p)-f(k))\right)}{(p-k)^2}=0\,.}
\end{lem}

\begin{proof}
The amplitude is in $L^1(D)$. If $f$ is constant, the assertion is immediate. Otherwise, real analyticity implies that the zeros of $f'$ are isolated. Thus the gradient of $(p,k)\mapsto f(p)-f(k)$ vanishes only on a set of measure zero. The coarea formula pushes the amplitude forward to an $L^1(\RR)$ function, and the assertion follows from the Riemann--Lebesgue lemma.
\end{proof}

The contribution close to an endpoint is controlled by the following Taylor estimate.

\begin{lem}\label[lem]{lem:Taylor expansion}
For $a\in\RR$, $t\geq0$, and $x>0$ sufficiently small,
\eq{&\int_{y\in[-x,x]}\dif{y}\sin\left(t\left(f\left(a+\frac{y+x}{2}\right)-f\left(a+\frac{y-x}{2}\right)\right)\right)\\
&\qquad=2x\sin(tf'(a)x)+R(t,x),}
where
\eq{|R(t,x)|\leq\frac{2}{3}t\norm{f''}_{L^\infty([a-x,a+x])}x^3\,.}
\end{lem}

\begin{proof}
Since the sine is Lipschitz-$1$,
\eq{|R(t,x)|\leq t\int_{-x}^{x}\dif{y}\left|f\left(a+\frac{y+x}{2}\right)-f\left(a+\frac{y-x}{2}\right)-f'(a)x\right|\,.}
The expression inside the absolute value equals
\eq{\int_{-x/2}^{x/2}\left(f'\left(a+\frac{y}{2}+s\right)-f'(a)\right)\dif{s}\,.}
Its absolute value is at most
\eq{\norm{f''}_{L^\infty([a-x,a+x])}\int_{-x/2}^{x/2}\left|\frac{y}{2}+s\right|\dif{s}=\frac{1}{4}\norm{f''}_{L^\infty([a-x,a+x])}(x^2+y^2)\,.}
Integration in $y$ proves the claim.
\end{proof}

\begin{lem}\label[lem]{lem:endpoint contributions}
For each component $[a_i,b_i]$ of $\Sigma_{E_F}$,
\eq{&-\frac{1}{\pi}\lim_{t\to\infty}\int_{p\in[a_i,b_i]}\dif{p}\int_{k\leq a_i}\dif{k}\frac{\sin\left(t(f(p)-f(k))\right)}{(p-k)^2}=-\frac{1}{2}\sgn(f'(a_i)),\\
&-\frac{1}{\pi}\lim_{t\to\infty}\int_{p\in[a_i,b_i]}\dif{p}\int_{k\geq b_i}\dif{k}\frac{\sin\left(t(f(p)-f(k))\right)}{(p-k)^2}=\frac{1}{2}\sgn(f'(b_i))\,.}
\end{lem}

\begin{proof}
We prove the first identity. Set $x=p-k$ and $y=p+k-2a_i$. The Jacobian is $1/2$, and the part of the transformed domain with $0<x<b_i-a_i$ has $-x\leq y\leq x$. Choose $\delta_t=t^{-3/4}$. By \cref{lem:Taylor expansion}, the contribution from $0<x<\delta_t$ is
\eq{-\frac{1}{2\pi}\int_0^{\delta_t}\frac{2x\sin(tf'(a_i)x)+R(t,x)}{x^2}\dif{x}\,.}
The error is $\calO(t\delta_t^2)=o(1)$, whereas
\eq{-\frac{1}{\pi}\lim_{t\to\infty}\int_0^{\delta_t}\frac{\sin(tf'(a_i)x)}{x}\dif{x}=-\frac{1}{2}\sgn(f'(a_i))\,.}

It remains to show that the complement gives zero. Choose $0<\eta<b_i-a_i$ so small that $f'$ has constant sign and is bounded away from zero on $[a_i-\eta,a_i+\eta]$. On $\delta_t<x<\eta$, put
\eq{g(x,y)=f\left(a_i+\frac{y+x}{2}\right)-f\left(a_i+\frac{y-x}{2}\right)\,.}
Then $|\partial_xg|$ is bounded away from zero. Using $\ee^{\ii tg}=(\ii t\partial_xg)^{-1}\partial_x\ee^{\ii tg}$ and integrating by parts in $x$, including the two moving boundary terms $y=\pm x$, gives
\eq{\left|\int_{\delta_t}^{\eta}\frac{\dif{x}}{x^2}\int_{-x}^{x}\ee^{\ii tg(x,y)}\dif{y}\right|\leq\frac{C}{t\delta_t}=o(1)\,.}
Here all derivatives of $g$ which occur are uniformly bounded on the fixed compact domain. On the remaining domain $x\geq\eta$, \cref{lem:oscillatory integral} applies. This proves the first identity. The second follows in the same way with $x=p-k<0$; its main term is
\eq{\frac{1}{\pi}\lim_{t\to\infty}\int_{-\delta_t}^{0}\frac{\sin(tf'(b_i)x)}{x}\dif{x}=\frac{1}{2}\sgn(f'(b_i))\,.}
\end{proof}

\begin{proof}[Proof of \cref{thm:main theorem}]

If $\Sigma_{E_F}$ is empty, then $P_F=0$ and the assertion is immediate.

 Otherwise, \cref{lem:sharp response} gives trace-classness, and \cref{lem:trace formula,lem:endpoint contributions} give the limit. By \cref{lem:Fermi intervals},

\eq{\frac{1}{2}\sum_{i=1}^{N}\left(\sgn(f'(b_i))-\sgn(f'(a_i))\right)=\sum_{i=1}^{N}1=N\,.}
This is the Euler characteristic of $\Sigma_{E_F}$.
\end{proof}
\begin{prop}[Finite-time response]\label[prop]{prop:finite-time response}
Under the hypotheses of \cref{thm:main theorem}, set

$N:=\chi(\Sigma_{E_F})$ and recall that

\eql{\calR(t)=2\pi\ii\tr\left(P_F[\Lambda_t,\Lambda]\right).
\label{eq:finite-time response}}

Then $\calR:[0,\infty)\to\RR$ is continuous,
\eq{\calR(0)=0,\qquad \lim_{t\to\infty}\calR(t)=N,}
and
\eq{\calR(t)
=
-\frac{1}{\pi}
\int_{p\in\Sigma_{E_F}}\dif{p}
\int_{k\notin\Sigma_{E_F}}\dif{k}\,
\frac{\sin\left(t(f(p)-f(k))\right)}{(p-k)^2}\,.}
In particular, for every $t\geq0$,
\eq{\left|\calR(t)-N\right|
\geq
N-\frac{1}{\pi}
\int_{p\in\Sigma_{E_F}}\dif{p}
\int_{k\notin\Sigma_{E_F}}\dif{k}\,
\frac{
\min\left\{1,t|f(p)-f(k)|\right\}
}{(p-k)^2}\,.}
The integral on the right tends to zero as $t\to0^+$.

Suppose, in addition, that
\eq{\inf_{k\notin\Sigma_{E_F}}|f'(k)|>0,
\qquad
\sup_{k\notin\Sigma_{E_F}}
\frac{|f''(k)|}{|f'(k)|^2}<\infty\,.}
Then there exists a constant $C<\infty$, depending only on $f$ and $E_F$,
such that
\eq{\left|\calR(t)-N\right|\leq Ct^{-1/3},
\qquad t\geq1\,.}
\end{prop}

\begin{proof}
The integral representation was established in the proof of
\cref{lem:trace formula}. The inequality
\eq{|\sin s|\leq\min\{1,|s|\}}
gives
\eq{|\calR(t)|
\leq
\frac{1}{\pi}
\int_{p\in\Sigma_{E_F}}\dif{p}
\int_{k\notin\Sigma_{E_F}}\dif{k}\,
\frac{
\min\left\{1,t|f(p)-f(k)|\right\}
}{(p-k)^2}\,,}
and the stated lower bound follows from
$|\calR(t)-N|\geq N-|\calR(t)|$.

Near a Fermi point, $|f(p)-f(k)|\leq C|p-k|$, so the integrand is bounded
by $C/|p-k|$, which is integrable across the corresponding
two-dimensional corner. Away from the Fermi points, $(p-k)^{-2}$ is
integrable on
$\Sigma_{E_F}\times\Sigma_{E_F}^c$. Dominated convergence therefore proves
continuity and $\calR(0)=0$.

For the quantitative large-time estimate, repeat the proof of
\cref{lem:endpoint contributions} with
\eq{\delta_t=t^{-2/3}.}
At each Fermi point, the Taylor remainder from
\cref{lem:Taylor expansion} contributes at most
$Ct\delta_t^2$, while the tail of the Dirichlet integral and the
intermediate region contribute at most $C/(t\delta_t)$. Hence the total
endpoint error is bounded by
\eq{C\left(t\delta_t^2+\frac{1}{t\delta_t}\right)
\leq Ct^{-1/3}\,.}

On the complement of the endpoint neighborhoods, the additional
hypotheses permit one integration by parts in $k$, using
\eq{\ee^{-\ii tf(k)}
=
-\frac{1}{\ii t f'(k)}
\partial_k\ee^{-\ii tf(k)}\,.}
The resulting amplitude and its derivative are integrable, and this part
is therefore $\calO(t^{-1})$. Summing over the finitely many Fermi points
proves the result.
\end{proof}

\begin{rem}[Necessity of the large-time limit]
The limit $t\to\infty$ in \cref{thm:main theorem} is essential rather than
a technical convenience. If $N>0$, then \cref{prop:finite-time response}
shows that the response evolves continuously from $\calR(0)=0$ to its
quantized value $N$. It must therefore assume noninteger values at finite
times, although accidental integer values are not excluded. Only
the large-time limit eliminates the non-topological transient
contributions.

Under the hypotheses of \cref{thm:main theorem} alone there is no
model-independent algebraic convergence rate. Real-analytic dispersions
may have stationary points of arbitrarily high finite order away from the
Fermi points, leading to arbitrarily slow stationary-phase decay. The
additional assumptions in \cref{prop:finite-time response} exclude these
stationary contributions and yield the displayed $t^{-1/3}$ bound. They
also force a nonempty Fermi sea to be a single interval: Rolle's theorem
gives a zero of $f'$ in any gap between two components.
\end{rem}

\begin{example}[Quadratic dispersion]\label[example]{ex:finite-time quadratic}
Consider
\eq{f(p)=p^2,\qquad E_F=1\,.}
Then
\eq{\Sigma_{E_F}=[-1,1],\qquad \chi(\Sigma_{E_F})=1\,.}
The finite-time response from \cref{eq:finite-time response} is
\eq{\calR(t)
=
-\frac{1}{\pi}
\int_{-1}^{1}\dif{p}
\int_{\RR\setminus[-1,1]}\dif{k}\,
\frac{\sin\left(t(p^2-k^2)\right)}{(p-k)^2}\,.}
Using the symmetry between the two components of
$\RR\setminus[-1,1]$, followed by the change of variables
$x=p-k$, $y=p+k$, gives
\eq{\calR(t)
=
\frac{2}{\pi t}
\int_0^\infty\dif{x}\,
\frac{\sin(tx^2)\sin(2tx)}{x^3}
=
\frac{2}{\pi}
\int_0^\infty\dif{u}\,
\frac{\sin(u^2)\sin(2\sqrt{t}\,u)}{u^3}\,.}

This formula gives a quantitative failure of finite-time quantization. In
fact,
\eq{\left|
\calR(t)-2\sqrt{\frac{2t}{\pi}}
\right|
\leq
\frac{4}{3\sqrt{\pi}}t^{3/2}\,.}
To see this, set
\eq{J(s):=\int_0^\infty\dif{u}\,
\frac{\sin(u^2)\sin(su)}{u^3}\,.}
Differentiating after inserting an Abel factor and then removing it gives
\eq{J(0)=J''(0)=0,\qquad
J'(0)=\int_0^\infty\dif{u}\,\frac{\sin(u^2)}{u^2}
=\sqrt{\frac{\pi}{2}}}
and
\eq{J'''(s)
=
-\int_0^\infty\dif{u}\,\sin(u^2)\cos(su)
=
-\frac{\sqrt{\pi}}{2}
\sin\left(\frac{\pi}{4}-\frac{s^2}{4}\right).}
Taylor's theorem therefore yields
\eq{\left|J(s)-s\sqrt{\frac{\pi}{2}}\right|
\leq\frac{\sqrt{\pi}}{12}s^3\,.}
Substituting $s=2\sqrt{t}$ and using
$\calR(t)=2J(2\sqrt{t})/\pi$ proves the displayed estimate.

For example, at $t=10^{-2}$,
\eq{0.1588<\calR(10^{-2})<0.1604\,,}
which is manifestly not quantized, whereas
\eq{\lim_{t\to\infty}\calR(t)=1}
by \cref{thm:main theorem}. Thus even for the elementary free-particle
dispersion, the integer arises only after the large-time limit.
\end{example}

\section{The discrete one-dimensional theory}\label{sec:discrete}

We now consider the corresponding problem on the lattice. Let
\eq{\calH=\ell^2(\ZZ),\qquad (X\psi)(n)=n\psi(n),}
and let $H$ be a translation-invariant Hamiltonian. Under the Fourier
transform
\eq{(\calF\psi)(p)=\sum_{n\in\ZZ}\psi(n)\ee^{-\ii np},
\qquad
\calF:\ell^2(\ZZ)\longrightarrow
L^2\left(\TT,\frac{\dif{p}}{2\pi}\right),
\qquad
\TT=\RR/(2\pi\ZZ),}
the Hamiltonian is multiplication by a real-valued function
$f:\TT\to\RR$:
\eq{(\calF H\calF^*\psi)(p)=f(p)\psi(p)\,.}
The spatial switch is the projection
\eq{(\Lambda(X)\psi)(n)=\chi_{\NN_0}(n)\psi(n)\,.}

We use the same occupation regularization $\rho_\eta(H)$ and response
$\calR_\eta(t)$ as in \cref{eq:regularized Kubo response}.

For a Fermi energy $E_F$, set
\eq{\Sigma_{E_F}:=\Set{p\in\TT\mid f(p)\leq E_F}}
and let $P_F=\chi_{(-\infty,E_F)}(H)$. Since a regular level set has measure
zero, $\calF P_F\calF^*$ agrees almost everywhere with multiplication by
$\chi_{\Sigma_{E_F}}$.

\begin{thm}[The lattice response]\label{thm:lattice response}
Let $f:\TT\to\RR$ be real analytic and suppose that $E_F$ is a regular value
of $f$. Then, for every fixed $t$, the operator
\eql{P_F\left[
\ee^{\ii tH}\Lambda(X)\ee^{-\ii tH},\Lambda(X)
\right]P_F}
is trace-class. Moreover,

\eql{&\lim_{t\to\infty}
2\pi\ii\tr\left(
P_F\left[
\ee^{\ii tH}\Lambda(X)\ee^{-\ii tH},\Lambda(X)
\right]P_F
\right)\notag\\
&\qquad=
\lim_{t\to\infty}\lim_{\eta\to0^+}\calR_\eta(t)
=\chi(\Sigma_{E_F})\,.
\label{eq:lattice response}}

\end{thm}

If $\Sigma_{E_F}$ is a proper nonempty subset of $\TT$, its Euler
characteristic is its number of connected components. The exceptional cases
satisfy
\eq{\chi(\varnothing)=0,\qquad \chi(\TT)=0\,.}
Thus a completely filled lattice band contributes zero rather than one.

We first record the geometry of the lattice Fermi sea.

\begin{lem}\label[lem]{lem:lattice Fermi arcs}
If $\Sigma_{E_F}$ is a proper nonempty subset of $\TT$, then, after choosing
a cut in its complement and lifting to $\RR$,
\eq{\Sigma_{E_F}=\bigsqcup_{j=1}^{N}[a_j,b_j],}
where
\eq{f'(a_j)<0<f'(b_j)\,.}
In particular, $\chi(\Sigma_{E_F})=N$.
\end{lem}

\begin{proof}
Since $E_F$ is a regular value, the points of $f^{-1}(E_F)$ are isolated.
Compactness of $\TT$ therefore makes their number finite. At the left
endpoint of each component, $f-E_F$ crosses from positive to negative,
whereas at the right endpoint it crosses from negative to positive. This
gives the asserted signs.
\end{proof}

On the lattice, the commutators with the time evolution are trace-class.

\begin{lem}\label[lem]{lem:lattice trace class}
For every fixed $t\in\RR$,
\eq{\left[\Lambda(X),\ee^{\ii tH}\right]\in\calJ_1.}
Consequently,
\eq{\left[
\ee^{\ii tH}\Lambda(X)\ee^{-\ii tH},\Lambda(X)
\right]\in\calJ_1.}
\end{lem}

\begin{proof}
The Fourier coefficients of the real-analytic periodic function
$\ee^{\ii tf}$ decay exponentially. Expanding $\ee^{\ii tH}$ in lattice
translations gives
\eq{\norm{\left[\Lambda(X),\ee^{\ii tH}\right]}_1
\leq
\sum_{r\in\ZZ}|r|
\left|\widehat{\ee^{\ii tf}}(r)\right|
<\infty\,.}
Indeed, the commutator of $\Lambda(X)$ with translation by $r$ sites has
rank $|r|$ and trace norm $|r|$.

It follows that
\eq{\ee^{\ii tH}\Lambda(X)\ee^{-\ii tH}-\Lambda(X)
=
\left[\ee^{\ii tH},\Lambda(X)\right]\ee^{-\ii tH}}
is trace-class. Hence
\eq{\left[
\ee^{\ii tH}\Lambda(X)\ee^{-\ii tH},\Lambda(X)
\right]
=
\left[
\ee^{\ii tH}\Lambda(X)\ee^{-\ii tH}-\Lambda(X),
\Lambda(X)
\right]}
is trace-class as well.
\end{proof}

The occupation regularization also gives a direct calculation of the sharp
lattice response.

\begin{lem}[The lattice trace formula]\label[lem]{lem:lattice trace formula}
For every $\eta>0$ and $t\geq0$,
\eql{\calR_\eta(t)
&=-\frac{1}{2\pi}
\int_{\TT}\dif{p}\int_{\TT}\dif{k}\,
\frac{\left(\rho_\eta(f(p))-\rho_\eta(f(k))\right)
\sin\left(t(f(p)-f(k))\right)}
{\left|\ee^{\ii p}-\ee^{\ii k}\right|^2}\,.
\label{eq:regularized lattice trace}}
Moreover,
\eql{\lim_{\eta\to0^+}\calR_\eta(t)
&=2\pi\ii\tr\left(
P_F\left[
\ee^{\ii tH}\Lambda(X)\ee^{-\ii tH},\Lambda(X)
\right]P_F
\right)\notag\\
&=-\frac{1}{\pi}
\int_{p\in\Sigma_{E_F}}\dif{p}
\int_{k\notin\Sigma_{E_F}}\dif{k}\,
\frac{\sin\left(t(f(p)-f(k))\right)}
{\left|\ee^{\ii p}-\ee^{\ii k}\right|^2}\,.
\label{eq:sharp lattice trace}}
\end{lem}

\begin{proof}
Since $\rho_\eta\circ f$ is smooth and periodic, its Fourier coefficients
satisfy
\eq{\sum_{r\in\ZZ}|r|\left|\widehat{\rho_\eta\circ f}(r)\right|<\infty\,.}
The translation argument in \cref{lem:lattice trace class} therefore shows
that
\eq{[\Lambda,\rho_\eta(H)]\in\calJ_1\,.}
Also, $\Lambda_t-\Lambda\in\calJ_1$ by
\cref{lem:lattice trace class}.
The operator $[\Lambda,\rho_\eta(H)]\Lambda$ is off-diagonal with
respect to the projection $\Lambda$, so its trace is zero. Cyclicity
involving the trace-class factor $\Lambda_t-\Lambda$ now gives
\eq{\frac{\calR_\eta(t)}{2\pi\ii}
&=\tr\left([\Lambda,\rho_\eta(H)](\Lambda_t-\Lambda)\right)\\
&=\tr\left(\rho_\eta(H)[\Lambda_t,\Lambda]\right)\,.}

With respect to the measure $\dif{k}/(2\pi)$, the ordinary kernels are
\eql{[\Lambda,\rho_\eta(H)](p,k)
&=\frac{\rho_\eta(f(k))-\rho_\eta(f(p))}
{1-\ee^{-\ii(p-k)}},\notag\\
(\Lambda_t-\Lambda)(p,k)
&=\frac{\ee^{\ii t(f(p)-f(k))}-1}
{1-\ee^{-\ii(p-k)}}\,.\label{eq:lattice difference kernel}}
Indeed, the differences in the numerators remove the diagonal singularity
of the Hardy kernel. Both kernels extend smoothly across $p=k$.
Computing the trace of their product and cancelling the cosine term by
$p\leftrightarrow k$ proves \cref{eq:regularized lattice trace}.

Since $\rho_\eta(H)\to P_F$ strongly and these operators are uniformly
bounded, the trace-classness of $[\Lambda_t,\Lambda]$ proves the first
equality in \cref{eq:sharp lattice trace}. To pass to the limit in the
integral, the empty and full Fermi seas are immediate, since
$\rho_\eta\circ f$ is then constant for all sufficiently small $\eta$.
Otherwise, for small $\eta$, its transitions lie in the monotone
neighborhoods of the $2N$ Fermi points, and
\eq{\norm{(\rho_\eta\circ f)'}_{L^1(\TT)}=2N\,.}
Thus
\eq{\int_{\TT}\dif{p}\,
\left|\rho_\eta(f(p))-\rho_\eta(f(p-h))\right|
\leq 2N|h|\,.}
Together with
$|\sin(t(f(p)-f(p-h)))|\leq C_t|h|$ and
$|1-\ee^{\ii h}|^2\geq ch^2$ for $|h|\leq\pi$, this bounds the integral
in $p$ of the absolute integrand uniformly in $\eta$ near $h=0$.
Dominated convergence first in $p$ for $h\neq0$ and then in $h$ proves the
limit. The weight becomes
$\chi_{\Sigma_{E_F}}(p)-\chi_{\Sigma_{E_F}}(k)$; exchanging $p$ and $k$ in
one of the two cross terms proves the second equality in
\cref{eq:sharp lattice trace}.
\end{proof}

It remains to evaluate the large-time limit.

\begin{lem}\label[lem]{lem:lattice endpoint contributions}
Suppose that
$\Sigma_{E_F}=\bigsqcup_{j=1}^{N}[a_j,b_j]$ as in
\cref{lem:lattice Fermi arcs}. Then the two endpoints of the $j$-th Fermi
arc contribute
\eq{-\frac{1}{2}\sgn(f'(a_j))
+\frac{1}{2}\sgn(f'(b_j))=1}
to the large-time limit in \cref{eq:sharp lattice trace}.
\end{lem}

\begin{proof}
Away from the diagonal $p=k$ modulo $2\pi$, the amplitude in
\cref{eq:sharp lattice trace} is integrable and smooth. The argument of
\cref{lem:oscillatory integral}, applied in local coordinates on $\TT^2$,
shows that this part tends to zero as $t\to\infty$.

Near an endpoint, choose lifts of $p$ and $k$ to $\RR$. Then
\eq{\frac{1}{|\ee^{\ii p}-\ee^{\ii k}|^2}
=
\frac{1}{4\sin^2((p-k)/2)}
=
\frac{1}{(p-k)^2}+\calO(1)\,.}
The contribution of the bounded remainder tends to zero by the
Riemann--Lebesgue lemma. The singular term is therefore identical to the
continuum endpoint calculation in \cref{lem:endpoint contributions}. Thus
the left and right endpoints contribute, respectively,
\eq{-\frac{1}{2}\sgn(f'(a_j)),
\qquad
\frac{1}{2}\sgn(f'(b_j))\,.}
The signs in \cref{lem:lattice Fermi arcs} show that each contribution is
$1/2$, proving the claim.
\end{proof}

\begin{proof}[Proof of \cref{thm:lattice response}]

If $\Sigma_{E_F}=\varnothing$, then $P_F=0$ and both responses vanish. If
$\Sigma_{E_F}=\TT$, then $P_F=I$. For the sharp response,
$\Lambda_t-\Lambda\in\calJ_1$ by
\cref{lem:lattice trace class}, so
\eq{\tr([\Lambda_t,\Lambda])
=\tr([\Lambda_t-\Lambda,\Lambda])=0}
by cyclicity. The occupation-regularized response vanishes for sufficiently
small $\eta$, since then $\rho_\eta(H)=I$. Thus both limiting responses
vanish, in agreement with $\chi(\TT)=0$.

In the remaining case, combine
\cref{lem:lattice trace formula,lem:lattice endpoint contributions} to obtain
\eq{\lim_{t\to\infty}
2\pi\ii\tr\left(
P_F\left[
\ee^{\ii tH}\Lambda(X)\ee^{-\ii tH},\Lambda(X)
\right]P_F
\right)
=
\frac{1}{2}\sum_{j=1}^{N}
\left(
\sgn(f'(b_j))-\sgn(f'(a_j))
\right)
=N\,.}

By \cref{lem:lattice Fermi arcs}, this is $\chi(\Sigma_{E_F})$. The equality
with the occupation-regularized response follows from
\cref{lem:lattice trace formula}.

\end{proof}

    \bigskip
	\bigskip
	\noindent\textbf{Acknowledgments.} 
JS was supported in part by NSF grant DMS-2510207 and the “ChatGPT for Academic Researchers program” of OpenAI. We are indebted to Charles Kane, Martin Fraas, Gian Michele Graf, Kohei Kawabata, Shinsei Ryu, and Pok Man Tam for stimulating discussions.
	\bigskip
	\appendix
\section{Numerical demonstrations}\label{sec:numerics}

We illustrate the finite-volume approximation of the lattice response and
separate the thermodynamic limit from the large-time limit. This distinction
is necessary because, at every fixed finite volume, the response is a finite
sum of trigonometric functions and is therefore quasiperiodic. In particular,
the limit $t\to\infty$ generally does not exist before the thermodynamic limit
has been taken.

\subsection{Finite-volume prescription}

Let
\eq{\calH_L:=\ell^2(\{-L,\ldots,L\})}
and impose Dirichlet boundary conditions. We consider the finite-volume
Hamiltonian
\eq{(H_{0,L}\psi)_n
=2\psi_n-\psi_{n+1}-\psi_{n-1}\,.}
For the clean calculation below, set $H_L:=H_{0,L}$.
The corresponding infinite-volume dispersion relation is
\eq{f(p)=2-2\cos(p),\qquad p\in\TT\,.}
For $E_F=1$, the Fermi sea is
\eq{\Sigma_{E_F}=[-\pi/3,\pi/3]\subset\TT}
and hence
\eq{\chi(\Sigma_{E_F})=1\,.}

We use the finite-volume switch and Fermi projection
\eq{\Lambda_L:=\chi_{\{0,\ldots,L\}}(X),
\qquad
P_{F,L}:=\chi_{(-\infty,E_F)}(H_L)}
and define
\eql{\calR_L(t):=
2\pi\ii\tr\left(
P_{F,L}\left[
\ee^{\ii tH_L}\Lambda_L\ee^{-\ii tH_L},
\Lambda_L
\right]P_{F,L}
\right).
\label{eq:finite-volume numerical response}}

At each fixed $L$, the occupation $\rho_\eta(H_L)$ equals $P_{F,L}$
for sufficiently small $\eta$. Thus this computes the zero-smoothing
response, whose infinite-volume expression is identified in
\cref{lem:lattice trace formula}.

If $H_L\varphi_a=E_a\varphi_a$, expansion in a complete eigenbasis gives
\eq{\calR_L(t)
=
4\pi
\sum_{\substack{E_a<E_F\\E_b\geq E_F}}
\sin\left(t(E_b-E_a)\right)
\left|\braket{\varphi_a|\Lambda_L|\varphi_b}\right|^2\,.}
All occupied eigenvectors determine $P_{F,L}$. Evaluating the displayed
response in an eigenbasis also requires the occupied-to-unoccupied
transition matrix elements, so the code below uses the complete eigensystem.

The infinite-volume theorem motivates testing whether
\eq{\lim_{t\to\infty}\lim_{L\to\infty}\calR_L(t)=1\,,}
with the thermodynamic limit taken first. For the numerical illustration,
we use the diagonal sequence
\eq{t_L=L^{1/2}\,.}
It satisfies
\eq{t_L\to\infty,\qquad \frac{t_L}{L}\to0,}
so the observation time diverges while remaining shorter than the time
required for a disturbance to reach the boundary. The maximal group velocity
for the present dispersion is
\eq{\max_{p\in\TT}|f'(p)|=2\,.}

\begin{figure}[ht]
\centering
\begin{tikzpicture}[x=0.037cm,y=12cm]
    \draw[->] (0,0.78) -- (305,0.78) node[right] {$L$};
    \draw[->] (0,0.78) -- (0,1.07)
        node[above] {$\calR_L(L^{1/2})$};

    \foreach \x in {0,50,100,150,200,250,300}{
        \draw (\x,0.78) -- (\x,0.775)
            node[below] {\small $\x$};
    }
    \foreach \y in {0.8,0.9,1.0}{
        \draw (0,\y) -- (-4,\y)
            node[left] {\small $\y$};
    }

    \draw[gray,densely dashed] (0,1) -- (300,1);

    \draw[very thick,ct_blue] plot coordinates {
        (12,0.82328)
        (18,0.93581)
        (27,0.93899)
        (36,0.96787)
        (54,1.02060)
        (72,1.02848)
        (108,0.99545)
        (144,1.00411)
        (216,1.01332)
        (288,0.98677)
    };

    \foreach \x/\y in {
        12/0.82328,
        18/0.93581,
        27/0.93899,
        36/0.96787,
        54/1.02060,
        72/1.02848,
        108/0.99545,
        144/1.00411,
        216/1.01332,
        288/0.98677
    }{
        \fill[ct_blue] (\x,\y) circle[radius=1.6pt];
    }
\end{tikzpicture}
\caption{The clean open-chain response evaluated along
$t_L=L^{1/2}$. The dashed line is
$\chi(\Sigma_{E_F})=1$. The residual oscillation reflects finite-time and
finite-volume effects.}
\label{fig:clean thermodynamic sequence}
\end{figure}
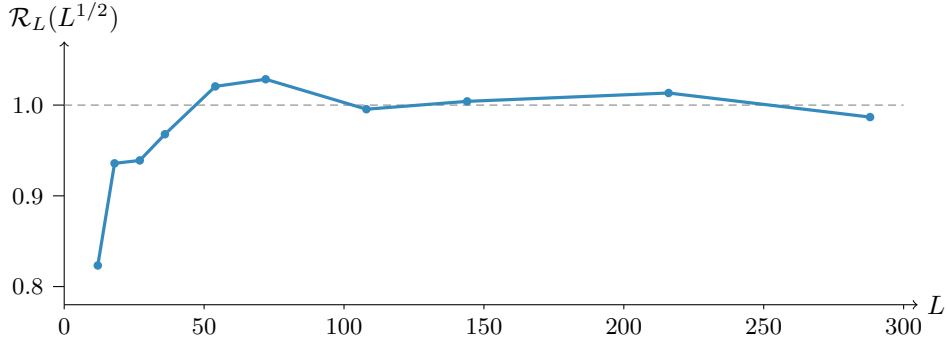

\subsection{Boundary conditions}

Open boundary conditions provide the most direct approximation of the
half-line switch. Periodic boundary conditions require additional care.
Consider the ring
\eq{\calH_L^{\mathrm{per}}
:=\ell^2\left(\ZZ/(2L\ZZ)\right)}
and the half-ring projection
\eq{\Lambda_L^{\mathrm{per}}
:=\chi_{\{0,\ldots,L-1\}}(X)\,.}
Unlike the half-line switch, $\Lambda_L^{\mathrm{per}}$ has two interfaces:
one at the bond $(-1,0)$ and another at $(L-1,L)$.

The contributions from the two interfaces add rather than cancel. Indeed,
\eq{\left[
U(1-\Lambda)U^*,1-\Lambda
\right]
=
[U\Lambda U^*,\Lambda]\,.}
Consequently, before the two interfaces communicate,
\eq{\calR_L^{\mathrm{per}}(t)
\simeq 2\calR_L(t)\,.}
The quantity which should be compared with the single-interface open-chain
response is therefore
\eq{\frac{1}{2}\calR_L^{\mathrm{per}}(t)\,.}

\begin{figure}[ht]
\centering
\begin{tikzpicture}[x=0.36cm,y=3cm]
    \draw[->] (0,0) -- (31.5,0) node[right] {$t$};
    \draw[->] (0,0) -- (0,1.45) node[above] {$\calR(t)$};

    \foreach \x in {0,5,10,15,20,25,30}{
        \draw (\x,0) -- (\x,-0.025)
            node[below] {\small $\x$};
    }
    \foreach \y in {0.5,1.0}{
        \draw (0,\y) -- (-0.25,\y)
            node[left] {\small $\y$};
    }

    \draw[gray,densely dashed] (0,1) -- (30,1);

    \draw[very thick,ct_orange,densely dashed] plot coordinates {
      (0.0,0.00000) (0.5,1.35226) (1.0,1.32897)
      (1.5,0.95759) (2.0,1.14100) (2.5,1.18787)
      (3.0,0.92107) (3.5,0.90782) (4.0,1.00352)
      (4.5,0.93320) (5.0,0.93840) (5.5,1.01342)
      (6.0,0.98910) (6.5,1.00858) (7.0,1.06686)
      (7.5,1.01534) (8.0,0.98977) (8.5,1.04055)
      (9.0,1.00569) (9.5,0.95585) (10.0,0.99340)
      (10.5,0.99666) (11.0,0.96810) (11.5,0.99884)
      (12.0,1.01067) (12.5,0.99118) (13.0,1.01888)
      (13.5,1.02690) (14.0,0.99131) (14.5,1.00387)
      (15.0,1.01971) (15.5,0.98416) (16.0,0.98161)
      (16.5,1.00525) (17.0,0.98869) (17.5,0.98646)
      (18.0,1.00980) (18.5,1.00014) (19.0,0.99762)
      (19.5,1.01964) (20.0,1.00640) (20.5,0.99140)
      (21.0,1.01095) (21.5,1.00431) (22.0,0.98270)
      (22.5,0.99712) (23.0,1.00217) (23.5,0.98730)
      (24.0,0.99933) (24.5,1.00775) (25.0,0.99528)
      (25.5,1.00540) (26.0,1.01398) (26.5,0.99587)
      (27.0,0.99823) (27.5,1.01024) (28.0,0.99385)
      (28.5,0.98913) (29.0,1.00353) (29.5,0.99617)
      (30.0,0.99175)
    };

    \draw[very thick,ct_blue] plot coordinates {
      (0.0,0.00000) (0.5,1.35547) (1.0,1.33450)
      (1.5,0.96461) (2.0,1.14911) (2.5,1.19662)
      (3.0,0.92954) (3.5,0.91517) (4.0,1.00957)
      (4.5,0.93829) (5.0,0.94283) (5.5,1.01740)
      (6.0,0.99311) (6.5,1.01328) (7.0,1.07256)
      (7.5,1.02187) (8.0,0.99685) (8.5,1.04802)
      (9.0,1.01327) (9.5,0.96301) (10.0,0.99975)
      (10.5,1.00223) (11.0,0.97315) (11.5,1.00354)
      (12.0,1.01518) (12.5,0.99592) (13.0,1.02426)
      (13.5,1.03302) (14.0,0.99796) (14.5,1.01085)
      (15.0,1.02691) (15.5,0.99128) (16.0,0.98824)
      (16.5,1.01120) (17.0,0.99410) (17.5,0.99153)
      (18.0,1.01464) (18.5,1.00495) (19.0,1.00276)
      (19.5,1.02540) (20.0,1.01272) (20.5,0.99809)
      (21.0,1.01787) (21.5,1.01133) (22.0,0.98952)
      (22.5,1.00342) (23.0,1.00788) (23.5,0.99262)
      (24.0,1.00441) (24.5,1.01269) (25.0,1.00030)
      (25.5,1.01085) (26.0,1.01999) (26.5,1.00232)
      (27.0,1.00493) (27.5,1.01710) (28.0,1.00072)
      (28.5,0.99570) (29.0,1.00956) (29.5,1.00172)
      (30.0,0.99702)
    };

    \fill[white,opacity=0.85] (17.5,1.20)
        rectangle (29.5,1.40);
    \draw[very thick,ct_blue] (18.2,1.34) -- (20.8,1.34)
        node[right] {\small open};
    \draw[very thick,ct_orange,densely dashed]
        (18.2,1.25) -- (20.8,1.25)
        node[right] {\small $\frac12$ periodic};
\end{tikzpicture}
\caption{Comparison at $L=151$. The solid curve is the open-chain
single-interface response. The dashed curve is one half of the periodic
response, which contains two interfaces. The two agree closely throughout the
displayed pre-recurrence window.}
\label{fig:boundary condition comparison}
\end{figure}
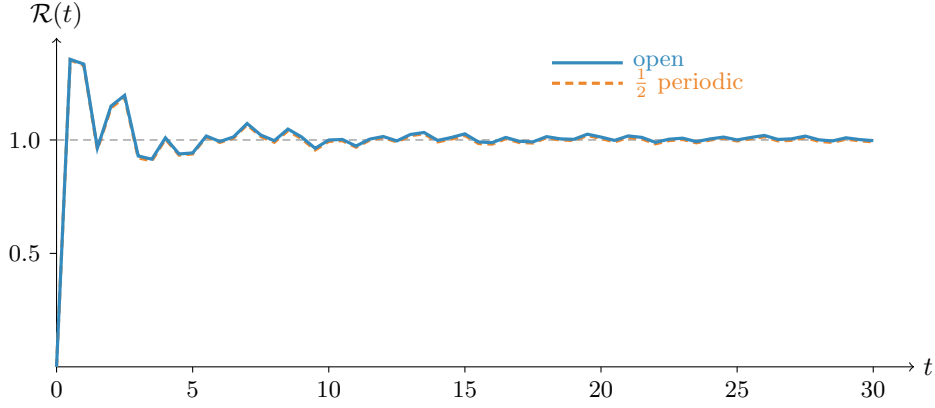

\subsection{The Anderson model}

We next consider
\eq{H_{\omega,L}
=
H_{0,L}
+
\sum_{n=-L}^{L}\omega_n\ket{n}\bra{n},
\qquad
\omega_n\sim\operatorname{Unif}[-W/2,W/2]}
with $W=3$. For every $W>0$, the one-dimensional Anderson model with a
non-degenerate iid uniform site potential has almost surely pure point
spectrum with exponentially localized eigenfunctions, and its static
conductivity vanishes
\cite{Kunz1980,frohlichConstructiveProofLocalization1985,Anderson_1958_PhysRev.109.1492,MottTwose1961,LeeRamakrishnan1985}.
Thus there is no reason to expect the clean ballistic response to
persist. 

A thermodynamic sequence must nevertheless be constructed before drawing
any conclusion. For every disorder sample, we first generate a potential on
the largest interval and obtain all smaller systems by restriction about the
same origin. Thus the values at different $L$ approximate one fixed infinite
disorder realization rather than unrelated samples.

For $M=200$ disorder realizations, define
\eq{m_L
:=
\frac{1}{M}\sum_{\alpha=1}^{M}
\calR_{\omega_\alpha,L}(L^{1/2})}
and
\eq{s_L
:=
\left(
\frac{1}{M-1}
\sum_{\alpha=1}^{M}
\left|
\calR_{\omega_\alpha,L}(L^{1/2})-m_L
\right|^2
\right)^{1/2}.}
A deterministic samplewise integer would require both
\eq{m_L\longrightarrow m\in\ZZ,
\qquad
s_L\longrightarrow0\,.}

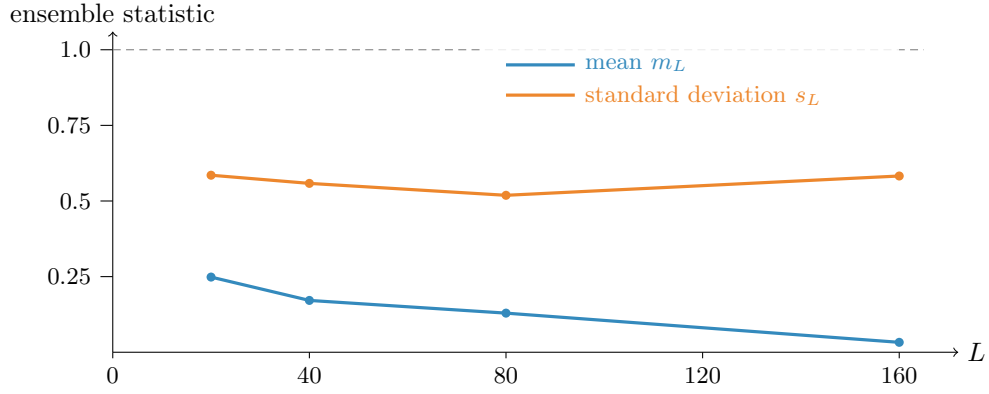
\begin{figure}[ht]
\centering
\begin{tikzpicture}[x=0.065cm,y=4cm]
    \draw[->] (0,0) -- (172,0) node[right] {$L$};
    \draw[->] (0,0) -- (0,1.06) node[above] {ensemble statistic};

    \foreach \x in {0,40,80,120,160}{
        \draw (\x,0) -- (\x,-0.02)
            node[below] {\small $\x$};
    }
    \foreach \y in {0.25,0.5,0.75,1.0}{
        \draw (0,\y) -- (-2.5,\y)
            node[left] {\small $\y$};
    }

    \draw[gray,densely dashed] (0,1) -- (165,1);

    \draw[very thick,ct_blue] plot coordinates {
        (20,0.24849)
        (40,0.17132)
        (80,0.12932)
        (160,0.03277)
    };
    \foreach \x/\y in {
        20/0.24849,
        40/0.17132,
        80/0.12932,
        160/0.03277
    }{
        \fill[ct_blue] (\x,\y) circle[radius=1.7pt];
    }

    \draw[very thick,ct_orange] plot coordinates {
        (20,0.58519)
        (40,0.55824)
        (80,0.51882)
        (160,0.58262)
    };
    \foreach \x/\y in {
        20/0.58519,
        40/0.55824,
        80/0.51882,
        160/0.58262
    }{
        \fill[ct_orange] (\x,\y) circle[radius=1.7pt];
    }

    \fill[white,opacity=0.85] (75,0.78)
        rectangle (160,1.02);
    \draw[very thick,ct_blue] (80,0.95) -- (94,0.95)
        node[right] {\small mean $m_L$};
    \draw[very thick,ct_orange] (80,0.85) -- (94,0.85)
        node[right] {\small standard deviation $s_L$};
\end{tikzpicture}
\caption{Anderson scaling at $W=3$, $E_F=1$, and
$t_L=L^{1/2}$, using $M=200$ nested disorder samples. The ensemble mean
moves toward zero, but the sample standard deviation does not decrease.
The data therefore give no evidence for a deterministic samplewise integer.
The dashed line denotes the clean value $1$.}
\label{fig:Anderson scaling}
\end{figure}

The displayed data are compatible with dephasing of the disorder-averaged
response toward zero. They do not establish such a limit. More importantly,
the variance remains of order one over the available sizes. Increasing the
volume removes the artificial boundary but does not spatially average the
disorder near the fixed switch interface. It is therefore possible for the
thermodynamic response to remain sample-dependent and oscillatory even when
its disorder average tends to zero.

\subsection{Numerical code}

The following code produces all data appearing in
\cref{fig:clean thermodynamic sequence,fig:boundary condition comparison,fig:Anderson scaling}.

\begin{verbatim}
import numpy as np


EF = 1.0


def open_hamiltonian(L, potential=None):
    """Dirichlet Hamiltonian on {-L,...,L}."""
    N = 2 * L + 1
    H = (
        2 * np.eye(N)
        - np.eye(N, k=1)
        - np.eye(N, k=-1)
    )
    if potential is not None:
        H = H + np.diag(potential)
    return H


def periodic_hamiltonian(L):
    """Clean Hamiltonian on Z/(2L Z)."""
    N = 2 * L
    H = (
        2 * np.eye(N)
        - np.eye(N, k=1)
        - np.eye(N, k=-1)
    )
    H[0, -1] = -1
    H[-1, 0] = -1
    return H


def open_switch(L):
    """Indicator of {0,...,L} in {-L,...,L}."""
    return (
        np.arange(2 * L + 1) >= L
    ).astype(float)


def periodic_switch(L):
    """Indicator of one half of the ring."""
    switch = np.zeros(2 * L)
    switch[:L] = 1.0
    return switch


def response(H, switch, times):
    """
    Return the finite-volume response at every time.

    This implementation uses the complete eigensystem.
    """
    times = np.atleast_1d(times)

    energies, vectors = np.linalg.eigh(H)
    occupied = energies < EF

    switch_eigenbasis = (
        vectors.conj().T
        @ (switch[:, None] * vectors)
    )

    gaps = (
        energies[~occupied][None, :]
        - energies[occupied][:, None]
    )
    weights = (
        4 * np.pi
        * np.abs(
            switch_eigenbasis[
                np.ix_(occupied, ~occupied)
            ]
        )**2
    )

    phases = np.sin(
        times[:, None, None] * gaps[None, :, :]
    )
    return np.einsum(
        "tab,ab->t", phases, weights
    )


# --------------------------------------------------
# 1. Clean thermodynamic sequence
# --------------------------------------------------

clean_sizes = np.array([
    12, 18, 27, 36, 54,
    72, 108, 144, 216, 288
])

clean_diagonal = np.empty(len(clean_sizes))

for j, L in enumerate(clean_sizes):
    time = np.sqrt(L)
    clean_diagonal[j] = response(
        open_hamiltonian(L),
        open_switch(L),
        np.array([time]),
    )[0]

print("clean diagonal sequence")
print(np.column_stack((
    clean_sizes,
    np.sqrt(clean_sizes),
    clean_diagonal,
)))


# --------------------------------------------------
# 2. Open versus periodic boundary conditions
# --------------------------------------------------

boundary_L = 151
boundary_times = np.arange(0.0, 30.1, 0.5)

open_curve = response(
    open_hamiltonian(boundary_L),
    open_switch(boundary_L),
    boundary_times,
)

periodic_curve = response(
    periodic_hamiltonian(boundary_L),
    periodic_switch(boundary_L),
    boundary_times,
)

# The periodic switch has two interfaces.
periodic_per_interface = periodic_curve / 2

print("boundary-condition comparison")
print(np.column_stack((
    boundary_times,
    open_curve,
    periodic_per_interface,
)))


# --------------------------------------------------
# 3. Anderson thermodynamic sequence
# --------------------------------------------------

disorder_strength = 3.0
anderson_sizes = np.array([20, 40, 80, 160])
number_of_samples = 200
largest_L = int(anderson_sizes[-1])

anderson_values = np.empty((
    number_of_samples,
    len(anderson_sizes),
))

for sample in range(number_of_samples):
    rng = np.random.default_rng(1000 + sample)

    # One realization on the largest interval.
    full_potential = rng.uniform(
        -disorder_strength / 2,
        disorder_strength / 2,
        2 * largest_L + 1,
    )

    for j, L_value in enumerate(anderson_sizes):
        L = int(L_value)

        # Restrict the same realization about the origin.
        potential = full_potential[
            largest_L - L:
            largest_L + L + 1
        ]

        anderson_values[sample, j] = response(
            open_hamiltonian(L, potential),
            open_switch(L),
            np.array([np.sqrt(L)]),
        )[0]

anderson_mean = anderson_values.mean(axis=0)
anderson_std = anderson_values.std(
    axis=0, ddof=1
)
anderson_standard_error = (
    anderson_std / np.sqrt(number_of_samples)
)

print("Anderson ensemble statistics")
print(np.column_stack((
    anderson_sizes,
    anderson_mean,
    anderson_std,
    anderson_standard_error,
)))
\end{verbatim}

For reference, the last block gives
\eq{
\begin{array}{c|cccc}
L & 20 & 40 & 80 & 160\\ \hline
m_L
&0.24849&0.17132&0.12932&0.03277\\
s_L
&0.58519&0.55824&0.51882&0.58262
\end{array}}
The standard errors of the four displayed means are, respectively,
$0.04138$, $0.03947$, $0.03669$, and $0.04120$.
  
  \begingroup
		\let\itshape\upshape
		\printbibliography
		\endgroup

    \end{document}